\documentclass[11pt,draftcls,onecolumn]{IEEEtran}
\newif\ifpdf
\ifx\pdfoutput\undefined
  \pdffalse
\else
  \ifnum\pdfoutput=1
    \pdftrue
  \else
    \pdffalse
  \fi
\fi
\usepackage{amsmath,epsfig, cite, cases}
\usepackage{amssymb,amscd,mathrsfs}
\usepackage[active]{srcltx}
\usepackage{color}
\usepackage{setspace}
\usepackage{enumerate}
\usepackage{bbding}
\usepackage{algorithm, algorithmic,subfigure}
\usepackage[center]{caption}

\newtheorem{thm}{Theorem}[section]
\newtheorem{cor}[thm]{Corollary}
\newtheorem{lem}[thm]{Lemma}

\newtheorem{rem}{Remark}
\newtheorem{define}{Definition}

\def\pc{H}
\def\gen{G}
\def\bh{\zeta}
\def\t{t}

\newcommand{\minW}[2]{w_{\min,#1}^{#2}}
\newcommand{\minWi}[1]{\minW{#1}{}}
\newcommand{\numW}[3]{N_{#1,#2}^{#3}}
\newcommand{\numWi}[2]{\numW{#1}{#2}{}}

\def\E{\mathbb{E}}

\newcommand{\argmin}{\operatornamewithlimits{argmin}}

\newcommand{\bra}[1]{\left(#1\right)}

\providecommand{\OO}[1]{\operatorname{O}\bigl(#1\bigr)}

\title{Linear Error Correcting Codes with Anytime Reliability}

\author{Ravi Teja Sukhavasi and Babak Hassibi
\thanks{Ravi Teja Sukhavasi is a graduate student with the department of Electrical Engineering, California Institute of Technology, Pasadena, USA
        {\tt\small teja@caltech.edu}}%
\thanks{Babak Hassibi is a faculty with the department of Electrical Engineering, California Institute of Technology,
        Pasadena, USA
        {\tt\small hassibi@caltech.edu}}%
\thanks{This work was supported in part by the National Science Foundation under grants CCF-0729203, CNS-0932428 and CCF-1018927, by the Office of Naval Research under the MURI grant N00014-08-1-0747, and by Caltech's Lee Center for Advanced Networking.}
}

\begin{document}

\maketitle
\thispagestyle{empty}
\pagestyle{empty}

\begin{abstract}
We consider rate $R=\frac{k}{n}$ causal linear codes that map a sequence of $k$-dimensional binary vectors $\{b_t\}_{t=0}^\infty$ to a sequence of $n$-dimensional binary vectors $\{c_t\}_{t=0}^\infty$, such that each $c_t$ is a function of $\{b_\tau\}_{\tau=0}^t$. Such a code is called anytime reliable, for a particular binary-input memoryless channel, if at each time instant $t$, and for all delays $d\geq d_o$, the probability of error $P\bra{\hat{b}_{t-d|t}\neq b_{t-d}}$ decays exponentially in $d$, i.e., $P\bra{\hat{b}_{t-d|t}\neq b_{t-d}} \leq \eta 2^{-\beta nd}$, for some $\beta > 0$. Anytime reliable codes are useful in interactive communication problems and, in particular, can be used to stabilize unstable plants across noisy channels. Schulman proved the existence of such codes which, due to their structure, he called tree codes in \cite{Schulman}; however, to date, no explicit constructions and tractable decoding algorithms have been devised. In this paper, we show the existence of anytime reliable \lq\lq linear\rq\rq\phantom{} codes with \lq\lq high probability\rq\rq, i.e., suitably chosen random linear causal codes are anytime reliable with high probability. The key is to consider time-invariant codes (i.e., ones with Toeplitz generator and parity check matrices) which obviates the need to union bound over all times. For the binary erasure channel we give a simple ML decoding algorithm whose average complexity is constant per time iteration and for which the probability that complexity at a given time $t$ exceeds $KC^3$ decays exponentially in $C$. We show the efficacy of the method by simulating the stabilization of an unstable plant across a BEC, and remark on the tradeoffs between the utilization of the communication resources and the control performance. 
\end{abstract}


\section{Introduction}
\label{sec: Introduction}
Shannon's information theory, in large part, is concerned with one-way communication of a message, that is available in its entirety, over a noisy communication network. There are increasingly many applications such as networked control systems \cite{ncs} and distributed computing \cite{Schulman} where communication is not one-way but interactive. A networked control system is characterized by the measurement unit begin separated from the control unit by a communication channel. The interactive nature of communication can be understood by observing that the measurements to be encoded are determined by the control inputs, which in turn, are determined by the encoded measurements received by the controller. Block encoding of the measurements is not applicable anymore because the controller needs real time information about the system so that an appropriate control input can be applied. This is especially critical when the system being controlled is open loop unstable. Any encoding-decoding delay translates into the system growing increasingly unstable. Hence, the desired reliability of communication is determined by the quality of the control input needed to stabilize the system. 

In the context of rate-limited deterministic channels, significant progress has been made (see eg.,\cite{Nair, Matveev}) in understanding the bandwidth requirements for stabilizing open loop unstable systems. When the communication channel is stochastic, \cite{Sahai} provides a necessary and sufficient condition on the communication reliability needed over a channel that is in the feedback loop of an unstable scalar linear process, and proposes the notion of anytime capacity as the appropriate figure of merit for such channels. In essense, the encoder is causal and the probability of error in decoding a source symbol that was transmitted $d$ time instants ago should decay exponentially in the decoding delay $d$. 

Although the connection between communication reliability and control is clear, very little is known about error-correcting codes that can achieve such reliabilities. Prior to the work of \cite{Sahai}, and in a different context, \cite{Schulman} proved the existence of codes which under maximum likelihood decoding achieve such reliabilities and referred to them as tree codes. Note that any real-time error correcting code is causal and since it encodes the entire trajectory of a process, it has a natural tree structure to it. \cite{Schulman} proves the existence of nonlinear tree codes and gives no explicit constructions and/or efficient decoding algorithms. Much more recently \cite{Ostrovsky} proposed efficient error correcting codes for unstable systems where the state grows only polynomially large with time.  So, for linear unstable systems that have an exponential growth rate, all that is known in the way of error correction is the existence of tree codes which are, in general, non-linear. When the state of an unstable scalar linear process is available at the encoder, \cite{Yuksel} and \cite{Simsek} develop encoding-decoding schemes that can stabilize such a process over the binary symmetric channel and the binary erasure channel respectively. But when the state is available only through noisy measurements, little is known in the way of stabilizing an unstable scalar linear process over a stochastic communication channel.

The subject of error correcting codes for control is in its relative infancy, much as the subject of block coding was after Shannon's seminal work in \cite{Shannon}. So, a first step towards realizing practical encoder-decoder pairs with anytime reliabilities is to explore linear encoding schemes. We consider rate $R=\frac{k}{n}$ causal linear codes which map a sequence of $k$-dimensional binary vectors $\{b_\tau\}_{\tau=0}^\infty$ to a sequence of $n-$dimensional binary vectors $\{c_\tau\}_{\tau=0}^\infty$ where $c_t$ is only a function of $\{b_\tau\}_{\tau=0}^t$. Such a code is anytime reliable if at all times $t$ and delays $d\geq d_o$, $P\bigl(\hat{b}_{t-d|t}\neq b_{t-d}\bigr) \leq \eta 2^{-\beta nd}$ for some $\beta > 0$. We show that linear tree codes exist and further, that they exist with a high probability. For the binary erasure channel, we propose a maximum likelihood decoder whose average complexity of decoding is constant per each time iteration and for which the probability that the complexity at a given time $t$ exceeds $KC^3$ decays exponentially in $C$. This allows one to stabilize a partially observed unstable scalar linear process over a binary erasure channel and to the best of the authors' knowledge, this has not been done before. 

In Section \ref{sec: Background}, we present some background and motivate the need for anytime reliability with a simple example. In Section \ref{sec: Sufficient Condition}, we come up with a sufficient condition for anytime reliability in terms of the weight distribution of the code. In Section \ref{sec: Existence}, we introduce the ensemble of time invariant codes and use the results from Section \ref{sec: Sufficient Condition} to prove that time invariant codes with anytime reliability exist with a high probability. In Section \ref{sec: DecodingBEC}, we present a simple decoding algorithm for the BEC and present simulations in Section \ref{sec: Simulations} to demonstrate the efficacy of the algorithm.
\section{Background and Problem Setup}
\label{sec: Background}
Owing to the duality between estimation and control, the essential complexity of stabilizing an unstable process over a noisy communication channel can be captured by studying the open loop estimation of the same process. So, we will illustrate the kind of communication reliability needed for control by analyzing the open loop estimation of the following random walk.

\textit{A toy example: }Consider estimating the following random walk, $x_{\t+1} = \lambda x_{\t} + w_\t$, where $w_t = \pm 1$ w.p $\frac{1}{2}$, $x_0 = 0$ and $|\lambda| > 1$. Suppose an observer observes $x_\t$ and communicates over a noisy communication channel to an estimator. The observer clearly needs to communicate one bit\footnote{If the channel is binary input, it should clearly allow at least one usage for each time step of the system evolution} telling whether $w_\t$ is $+1$ or $-1$. Now the estimator's estimate of the state, $\hat{x}_{\t+1|\t}$, is given by
\begin{align}
\label{eq: example}
 \hat{x}_{\t+1|\t} = \sum_{j=0}^\t\lambda_{\t-j}\hat{w}_{j|\t}
\end{align}
Suppose $P_{d,\t}^e = P\bra{\argmin_{j}(\hat{w}_{j|\t}\neq \hat{w}_{j}) = \t-d+1}$, i.e., the position of the earliest erroneous $\hat{w}_{j|t}$ is at time $j=t-d+1$. From \eqref{eq: example}, we can write $\E\bigl|x_{\t+1}-\hat{x}_{\t+1|\t}\bigr|^2$ as
\begin{align*}
\sum_{w_{0:\t},\hat{w}_{0:\t|\t}}P\bra{w_{0:\t},\hat{w}_{0:\t|\t}}\biggl| \sum_{j=1}^n\lambda^{\t-j}(w_j-\hat{w}_{j|\t})\biggr|^2 =\sum_{d\leq \t}P_{d,\t}^e \biggl| \sum_{j=\t-d+1}^\t\lambda^{\t-j}(w_j-\hat{w}_{j|\t})\biggr|^2 \leq \frac{4}{(|\lambda|-1)^2}\sum_{d\leq \t}P_{d,\t}^e|\lambda|^{2d}
\end{align*}
Clearly, a sufficient condition for $\limsup_\t\E\left|x_{\t+1}-\hat{x}_{\t+1|\t}\right|^2$ to be finite  is as follows
\begin{align}
 P_{d,\t}^e \leq |\lambda|^{-(2+\delta)d}\,\,\, \forall\,\,\, d\geq d_o,\,\,\, \t > \t_o\,\,\, \text{and } \delta > 0
\end{align}
where $d_o$ and $\t_o$ are constants that do no depend on $\t,d$. In the above example, noise is discrete valued. If it is not discrete, e.g., a uniform random variable on a bounded interval, then apart from the reliability criterion above, one would also need a minimum rate of communication that depends on the size of $\log_2|\lambda|$. But codes that communicate at a positive rate and satisfy the above reliability criterion have been more or less elusive. 
\begin{figure}[h]
\centering
 \includegraphics[scale=0.4]{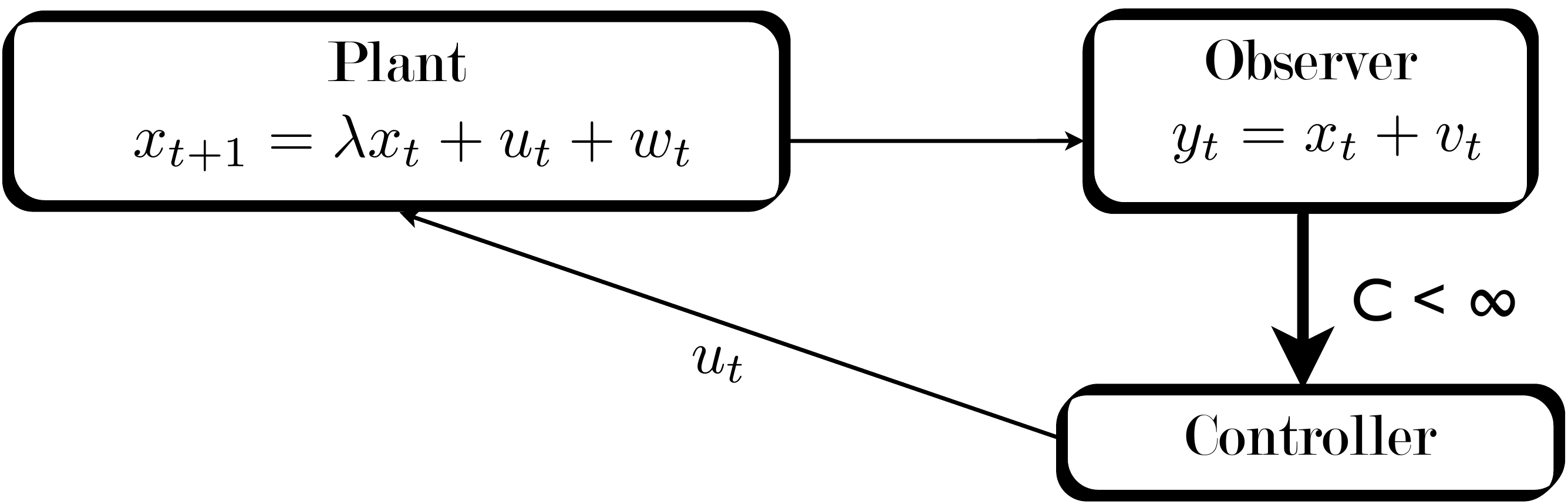}
\caption{System model}
\label{fig: system_model}
\end{figure}
Consider the following unstable scalar linear process
\begin{align}
\label{eq: system_model}
 x_{\t+1} = \lambda x_\t + u_\t + w_\t,\quad y_\t = x_\t + v_\t
\end{align}
where $|\lambda| > 1$, $u_\t$ is the control input and, $|w_\t| < \frac{W}{2}$ and $|v_\t| < \frac{V}{2}$ are bounded process and measurement noise variables. The measurements $\{y_\t\}$ are made by an observer while the control inputs $\{u_\t\}$ are applied by a remote controller that is connected to the observer by a noisy communication channel. Naturally, the measurements $y_{0:\t-1}$ will need to be encoded by the observer to provide protection from the noisy channel while the controller will need to decode the channel output to estimate the state $x_\t$ and apply a suitable control input $u_\t$. This can be accomplished by employing a channel encoder at the observer and a decoder at the controller. For simplicity, we will assume that the channel input alphabet is binary. Suppose one time step of system evolution in \eqref{eq: system_model} corresponds to $n$ channel uses\footnote{In practice, the system evolution in \eqref{eq: system_model} is obtained by discretizing a continuous time differential equation. So, the interval of discretization could be adjusted to correspond to an integer number of channel uses, provided the channel use instances are close enough.}. Then, at each instant of time $\t$,  the operations performed by the observer, the channel encoder,  the channel decoder and the controller can be described as follows. The observer generates a $k-$bit message, $b_\t\in\{0,1\}^k$, that is a causal function of the measurements, i.e., it depends only on $y_{0:\t}$. Then the channel encoder causally encodes $b_{0:\t} \in \{0,1\}^{k\t}$ to generate the $n$ channel inputs $c_\t\in\{0,1\}^n$. Note that the rate of the channel encoder is $R = k/n$. Denote the $n$ channel outputs corresponding to $c_\t$ by $z_\t \in \mathcal{Z}^n$, where $\mathcal{Z}$ denotes the channel output alphabet. Using the channel outputs received so far, i.e., $z_{0:\t}\in \mathcal{Z}^{n\t}$, the channel decoder generates estimates $\{\hat{b}_{\tau|\t}\}_{\tau \leq \t}$ of $\{b_\tau\}_{\tau\leq \t}$, which, in turn, the controller uses to generate the control input $u_{\t+1}$. This is illustrated in Fig. \ref{fig: operation}. Note that we do not assume any channel feedback. 
\begin{figure}
\centering
\includegraphics[scale=0.4]{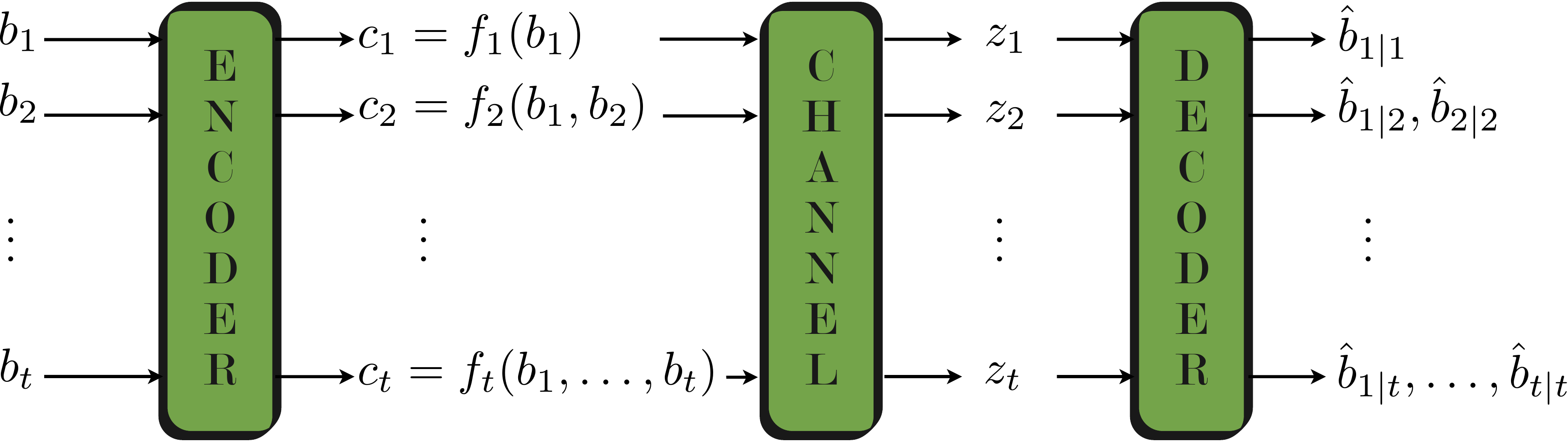}
\caption{Causal encoding and decoding}
\label{fig: operation}
\end{figure}
Then using the lattice quantizer argument presented in \cite{Sahai}, in the limit of large $n$ and large $\lambda$, we have the following sufficient condition on the performance of the encoder-decoder pair so that the unstable process in \eqref{eq: system_model} can be stabilized
\begin{lem}[Theorem 5.2 \cite{Sahai}]
\label{lem: Sahai}
 It is possible to control the unstable scalar process \eqref{eq: system_model} over a noisy communication channel so that $\limsup_t\E|x_t|^m < \infty$ if, for some rate $R > \frac{1}{n}\log_2|\lambda|$ and exponent $\beta > \frac{m}{n}\log_2|\lambda|$, we have
\begin{align*}
 P\bra{\min\{\tau:\hat{b}_{\tau|\t}\neq b_\tau\}=\t-d+1} \leq \eta 2^{-\beta nd},\,\forall\,\,d\geq d_o,\,t > t_o
\end{align*}
where $d_o$ and $t_o$ are constants independent of $d,t$.
\end{lem}
In what follows, we will demonstrate causal linear codes which under maximum likelihood decoding achieve such exponential reliabilities. 
\section{Linear Anytime Codes - A Sufficient Condition}
\label{sec: Sufficient Condition}
As discussed earlier, a first step towards developing practical encoding and decoding schemes for automatic control is to study the existence of linear codes with anytime reliability. We will begin by defining a causal linear code.
\begin{define}[Causal Linear Code]
 A causal linear code is a sequence of linear maps $f_\tau:\{0,1\}^{k\tau}\mapsto \{0,1\}^n$ and hence can be represented as
\begin{align}
 f_\tau(b_{1:\tau}) = \gen_{\tau 1}b_1 + \gen_{\tau 2}b_2 + \ldots + \gen_{\tau \tau}b_\tau
\end{align}
where $\gen_{ij}\in\{0,1\}^{n\times k}$
\end{define}
We denote $c_\tau \triangleq f_\tau(b_{1:\tau})$. Note that a tree code is a more general construction where $f_\tau$ need not be linear. Also note that the associated code rate is $R = \frac{k}{n}$.  One can alternately represent a causal linear code by an infinite dimensional block lower triangular generator matrix $\mathbb{G}_{n,R}$ or equivalently as an infinite dimensional block lower triangular parity check matrix, $\mathbb{H}_{n,R}$. 
\begin{align}
\label{eq: paritycheck}
 \mathbb{G}_{n,R} = \left[\begin{array}{ccccc}
				\gen_{11} & 0 & \ldots & \ldots & \ldots \\
				\gen_{21} & \gen_{22} & 0 & \ldots & \ldots \\
				\vdots & \vdots & \ddots & \vdots & \vdots \\
				\gen_{\tau 1} & \gen_{\tau 2} & \ldots & \gen_{\tau\tau} & 0\\
				\vdots & \vdots & \vdots & \vdots & \ddots
				\end{array}\right],\,\,\,\, \mathbb{H}_{n,R} = \left[\begin{array}{ccccc}
				\pc_{11} & 0 & \ldots & \ldots & \ldots \\
				\pc_{21} & \pc_{22} & 0 & \ldots & \ldots \\
				\vdots & \vdots & \ddots & \vdots & \vdots \\
				\pc_{\tau 1} & \pc_{\tau 2} & \ldots & \pc_{\tau\tau} & 0\\
				\vdots & \vdots & \vdots & \vdots & \ddots
				\end{array}\right]
\end{align}
where $\pc_{ij}\in\{0,1\}^{\overline{n}\times n}$ and $\overline{n}=n(1-R)$\footnote{While for a given generator matrix, the parity check matrix is not unique, when $G_{n,R}$ is block lower, it is easy to see that $H_{n,R}$ can also be chosen to be block lower.}. In fact, we present all our results in terms of the parity check matrix. Before proceeding further, it is useful to introduce some notation

\subsection{Notation}
 \begin{enumerate}
\item $\mathbb{H}_{n,R}^t \triangleq \overline{n}t\times nt\text{ leading principal minor of }\mathbb{H}_{n,R}$
\item $\mathcal{C}_t \triangleq \left\{c\in\{0,1\}^{nt}: \mathbb{H}_{n,R}^tc = 0\right\}$
\item $\mathcal{C}_{t,d} \triangleq \left\{c\in\mathcal{C}_t: c_{\tau<t-d+1}=0,\,\,c_{t-d+1}\neq 0\right\}$
\item $\numW{w}{d}{t} \triangleq \left|\{c\in\mathcal{C}_{t,d}:\|c\|=w\}\right|$
\item $\minW{d}{t} \triangleq \argmin_{w}(\numW{w}{d}{t}\neq 0)$
\item $P_{t,d}^e \triangleq  P\bra{\min\{\tau:\hat{b}_{\tau|\t}\neq b_\tau\}=\t-d+1}$
\label{eq: notation}
 \end{enumerate}
where $\|c\|$ denotes the Hamming weight of $c$.
\subsection{A Sufficient Condition}
The objective is to study the existence of causal linear codes which under ML decoding guarantee
\begin{align}
 P_{d,t}^e \leq \eta 2^{-\beta d},\,\,\,\forall \,\,\,t,\,\,d\geq d_o
\end{align}
where $d_o$ is a constant independent of $d,t$. In what follows, we will develop a sufficient condition for a linear code to be anytime reliable in terms of its weight distribution. Suppose the decoding instant is $\t$ and without loss of generality, assume that the all zero codeword is transmitted, i.e., $c_{\tau}=0$ for $\tau\leq \t$. We are interested in the error event where the earliest error in estimating $b_\tau$ happens at $\tau =\t-d+1$, i.e., $\hat{b}_{\tau|\t}=0$ for all $\tau < \t-d+1$ and $\hat{b}_{\t-d+1|\t}\neq 0$. Note that this is equivalent to the ML codeword, $\hat{c}$, satisfying $\hat{c}_{\tau < \t-d+1} = 0$ and $\hat{c}_{\t-d+1}\neq 0$, and $\mathbb{H}_{n,R}^t$ having full rank so that $\hat{c}$ can be uniquely mapped to a transmitted sequence $\hat{b}$. Then we have
\begin{subequations}
\begin{align}
P_{t,d}^e =  P\left[\bigcup_{c\in\mathcal{C}_{t,d}}0\text{ is decoded as }c\right] \leq \sum_{c\in\mathcal{C}_{t,d}}P\bra{0\text{ is decoded as }c}\label{eq: weight1}
\end{align}
\end{subequations}
Now, it is well known (for eg, see \cite{Shamai}) that, under maximum likelihood decoding, $P\bra{0\text{ is decoded as }c}\leq \bh^{\|c\|}$, where $\bh$ is the Bhattacharya parameter, i.e., 
\begin{align*}
 \bh = \int\limits_{-\infty}^{\infty}\sqrt{p(z|X=1)p(z|X=0)}dz
\end{align*}
where, $z$ and $X$ denote the channel output and input respectively. From \eqref{eq: weight1}, it follows that
\begin{align*}
 P_{t,d}^e \leq  \sum_{\minW{d}{t}\leq w\leq nd}\numW{w}{d}{t}\bh^w
\end{align*}
If $\minW{d}{t} \geq \alpha nd$ and $\numW{w}{d}{t}\leq 2^{\theta w}$ for some $\theta < \log_2(1/\bh)$, then
\begin{align}
\label{eq: sufficient_condition}
 P_{t,d}^e \leq \eta2^{-\alpha nd(\log_2(1/\bh) - \theta)}
\end{align}
where $\eta = (1-2^{\log_2(1/\bh) - \theta})^{-1}$. So, an obvious sufficient condition for $\mathbb{H}_{n,R}$ can be described in terms of $\minW{d}{t}$ and $\numW{w}{d}{t}$ as follows. For some $\theta < \log_2(1/\bh)$, we need
\begin{subequations}
 \label{eq: weight_distribution}
\begin{align}
 \minW{d}{t} \geq \alpha nd,\quad\numW{w}{d}{t} \leq 2^{\theta w}\,\,\,\forall\,\,\, t,\,\,\,d\geq d_o
\end{align}
where $d_o$ is a constant that is independent of $d,t$. This brings us to the following definition
\begin{define}[Anytime distance and Anytime reliability]
 We say that a code $\mathbb{H}_{n,R}$ has $(\alpha,\theta,d_o)-$anytime distance, if the following hold
\begin{enumerate}
 \item $\mathbb{H}_{n,R}^t$ is full rank for all $t>0$
 \item $\minW{d}{t} \geq \alpha nd$, $\numW{w}{d}{t} \leq 2^{\theta w}$ for all $t > 0$ and $d \geq d_o$.
\end{enumerate}
Also, we say that a code $\mathbb{H}_{n,R}$ is $(R,\beta,d_o)-$anytime reliable if, under ML decoding
\begin{align}
 P_{t,d}^e \leq \eta 2^{-\beta n d},\,\,\,\forall\,\,\,t >0,\,\,\,d\geq d_o
\end{align}
where, $d_o$ is a constant independent of $t,d$.
\end{define}
\end{subequations}

\section{Linear Anytime Codes - Existence}
\label{sec: Existence}

We will begin by proving the existence of such codes over a finite time horizon, $T$, i.e., $P_{d,t}^e \leq \eta 2^{-\beta d},\,\,\,\forall \,\,\,t\leq T,\,\,d\geq d_o$. We will then prove their existence for all time.

\subsection{Finite Time Horizon}
Over a finite time horizon, $T$, a causal linear code is represented by a block lower triangular parity check matrix $\mathbb{H}_{n,R,T}\in\{0,1\}^{\overline{n}T\times nT}$. The following Theorem guarantees the existence of a $\mathbb{H}_{n,R,T}$
 such that \eqref{eq: weight_distribution} is true for all $t\leq T$.
\begin{thm}[Appropriate Weight Distribution]
\label{thm: finiteTimeHorizon}
 For each time $T > 0$, rate $R > 0$, $\alpha < H^{-1}(1-R)$ and $\theta > \log_2(1/(2^{1-R}-1))$, there exists a causal linear code $H(n,k,T)$ that has $(\alpha,\theta,d_o)-$anytime distance, where $d_o$ is a constant independent of $d$, $t$ and $T$.
\end{thm}
$H^{-1}(1-R)$ is the smaller root of the equation $H(x)=1-R$, where $H(.)$ is the binary entropy function. The proof is by induction and is detailed in the Appendix. Theorem \ref{thm: finiteTimeHorizon} proves the existence of finite dimensional causal linear codes, $\mathbb{H}_{n,R,T}$, that are anytime reliable for decoding instants upto time $T$. In the following subsection, we demonstrate the existence of infinite dimensional causal linear codes, $\mathbb{H}_{n,R}$, that are anytime reliable for all decoding instants. We also show that such codes drawn from an appropriate ensemble are anytime reliable with a high probability.

\subsection{Time Invariant Codes}

Consider causal linear codes with the following Toeplitz structure
\begin{align*}
   \mathbb{H}_{n,R}^{TZ} = \left[\begin{array}{ccccc}
				\pc_{1} & 0 & \ldots & \ldots & \ldots \\
				\pc_{2} & \pc_{1} & 0 & \ldots & \ldots \\
				\vdots & \vdots & \ddots & \vdots & \vdots \\
				\pc_{\tau} & \pc_{\tau-1} & \ldots & \pc_{1} & 0\\
				\vdots & \vdots & \vdots & \vdots & \ddots
				\end{array}\right]
\end{align*}
The superscript $TZ$ in $\mathbb{H}_{n,R}^{TZ}$ denotes \lq Toeplitz'. $\mathbb{H}_{n,R}^{TZ}$ is obtained from $\mathbb{H}_{n,R}$ in \eqref{eq: paritycheck} by setting $\pc_{ij} = \pc_{i-j+1}$ for $i\geq j$. Due to the Toeplitz structure, we have the following invariance, $\minW{d}{t} = \minW{d}{t'}$ and $\numW{w}{d}{t} = \numW{w}{d}{t'}$ for all $t,t'$. The code $\mathbb{H}_{n,R}^{TZ}$ will be referred to as a time-invariant code. The notion of time invariance is analogous to the convolutional structure used to show the existence of infinite tree codes in \cite{Schulman}. This time invariance allows one to prove that such codes which are anytime reliable are abundant. 
\begin{define}[The ensemble $\mathbb{TZ}_p$]
 The ensemble $\mathbb{TZ}_p$ of time-invariant codes, $\mathbb{H}_{n,R}^{TZ}$, is obtained as follows, $\pc_1$ is any full rank binary matrix and for $\tau \geq 2$, the entries of $H_\tau$ are chosen i.i.d according to Bernoulli($p$), i.e., each entry is 1 with probability $p$ and 0 otherwise.
\end{define}
For the ensemble $\mathbb{TZ}_p$, we have the following result
\begin{thm}[Abundance of time-invariant codes]
 \label{thm: Toeplitz}
For each $R > 0$, $\alpha < H^{-1}\left[(1-R)\log_2(1/(1-p))\right]$ and $\theta > -\log_2\left[(1-p)^{-(1-R)}-1\right]$, we have
\begin{align}
 P\bra{\mathbb{H}_{n,R}^{TZ}\text{ has }(\alpha,\theta,d_o)-\text{anytime distance}} \geq 1-2^{-\Omega(nd_o)}
\end{align}
\end{thm}
We can now use this result to demonstrate an achievable region of rate-exponent pairs for a given channel, i.e., the set of rates $R$ and exponents $\beta$ such that one can guarantee $(R,\beta)$ anytime reliability using linear codes. To determine the values of $R$ that will satisfy \eqref{eq: sufficient_condition}, note that we need
\begin{align*}
 \log_2(1/(2^{1-R}-1)) < \log_2(1/\bh) \implies R < 1-\log_2(1+\bh)
\end{align*}
With this observation, we have the following Corollary.
\begin{cor}
\label{cor: thresholdsBEC}
For any rate $R$ and exponent $\beta$ such that
\begin{subequations}
\label{eq: thresholdsBEC}
\begin{align}
 R < 1- \frac{\log_{2}(1+\bh)}{\log_2(1/(1-p))}\quad\text{and}\quad\beta  < H^{-1}(1-R)\bra{\log_2\bra{\frac{1}{\bh}} + \log_2\bigl[(1-p)^{-(1-R)}-1\bigr]}
\end{align}
\end{subequations}
if $\mathbb{H}_{n,R}^{TZ}$ is chosen from $\mathbb{TZ}_p$, then
\begin{align*}
 P\bra{\mathbb{H}_{n,R}^{TZ}\text{ is }(R,\beta,d_o)-\text{anytime reliable}} \geq 1-2^{-\Omega(nd_o)}
\end{align*}
\end{cor}
Note that by choosing $p$ small, we can trade off better rates and exponents with sparser parity check matrices. 

\subsection{Tightening the Union Bound}
The region of achievable rate exponent pairs as described in Corollary \ref{cor: thresholdsBEC} have been obtained by using a simple union bound on the error probability. It is well known that for linear block codes, the union bound does not give an error exponent for rates close to the channel capacity. This appears to be the case in Corollary \ref{cor: thresholdsBEC} as well. For BSC($\epsilon$) with $p=\frac{1}{2}$, the threshold for rate in Corollary \ref{cor: thresholdsBEC} becomes $R < 1-2\log_2(\sqrt{\epsilon}+\sqrt{1-\epsilon})$. It turns out that there is a regime of bit flip probabilities, $\epsilon < \epsilon^*$, where this threshold can be improved as follows.
\begin{thm}[Tighter bounds for BSC($\epsilon$)]
\label{thm: tighterBSC}
Let $\epsilon^*$ be defined as
\begin{align}
  \epsilon^* \triangleq \argmin_{\epsilon > 0}\left\{1 - H(2\epsilon) \leq 1-2\log_2(\sqrt{\epsilon}+\sqrt{1-\epsilon})\right\}
 \end{align}
Consider BSC($\epsilon$) for $\epsilon < \epsilon^*$, then for any rate $R$ and exponent $\beta$ such that
\begin{align*}
 R < 1 - H(2\epsilon),\,\,\,\beta  < KL\bra{\frac{1}{2}H^{-1}(1-R)\| \epsilon}
\end{align*}
if $\mathbb{H}_{n,R}^{TZ}$ is chosen from $\mathbb{TZ}_{\frac{1}{2}}$, then
\begin{align*}
 P\bra{\mathbb{H}_{n,R}^{TZ}\text{ is }(R,\beta,d_o)-\text{anytime reliable}} \geq 1-2^{-\Omega(nd_o)}
\end{align*}
\end{thm}
A simple numerical calculation gives $\epsilon^* \approx 0.0753$. Hence for a binary symmetric channel with bit flip probability smaller than $0.0753$, one can get an anytime exponent for rates closer to the Shannon capacity than that suggested by a simple union bound. In general, given a linear block code with a specified weight distribution, there are numerous results in literature (e.g., see the survey \cite{Shamai}) that significantly tighten the union bound. Some of these techniques will directly yield improved thresholds than those obtained in Corollary \ref{cor: thresholdsBEC}. 

\subsection{Stabilizable Region}
Using the thresholds obtained in Corollary \ref{cor: thresholdsBEC} and Theorem \ref{thm: tighterBSC}, we can discuss the range of $|\lambda|$ for which the $m^{th}$ moment of $x_t$ in \eqref{eq: system_model} can be stabilized over some common channels. Using Lemma \ref{lem: Sahai}, an anytime reliable code with rate $R$ and exponent $\beta$ can stabilize the process in \eqref{eq: system_model} for all $\lambda$ such that 
\begin{align*}
 \log_2|\lambda| < \min\{nR,\frac{n\beta}{m}\}
\end{align*}
So, a scalar unstable linear process in \eqref{eq: system_model} can be stabilized over a channel with Bhattacharya parameter $\bh$ provided
\begin{align}
\label{eq: region1}
 \log_2|\lambda| < \log_2|\lambda_{max}| = \sup_{R < R_\bh,\beta < \beta_{\bh,R}}\min\{nR,\frac{\beta n}{m}\}
\end{align}
where $R_\bh = 1 - \log_2(1+\bh)$ and $\beta_{\bh,R} = H^{-1}(1-R)\bra{\log_2\bra{\frac{1}{\bh}} + \log_2\bra{2^{1-R}-1}}$ are obtained by setting $p=\frac{1}{2}$ in \eqref{eq: thresholdsBEC}. For the BSC($\epsilon$), using Theorem \ref{thm: tighterBSC}, one can tighten \eqref{eq: region1} by replacing $R_\bh$ with $1-H(\epsilon)$ and $\beta_{\bh,R}$ with $KL\bra{H^{-1}(1-R)\|\min\{\epsilon,1-\epsilon\}}$. For $m=2$, the stabilizable region for the BEC and BSC is shown in Fig \ref{fig: BECvBSC} where $|\lambda_{max}|^{\frac{1}{n}}$ is plotted against the channel parameter.
\begin{figure}
 \centering
\includegraphics[scale=0.4]{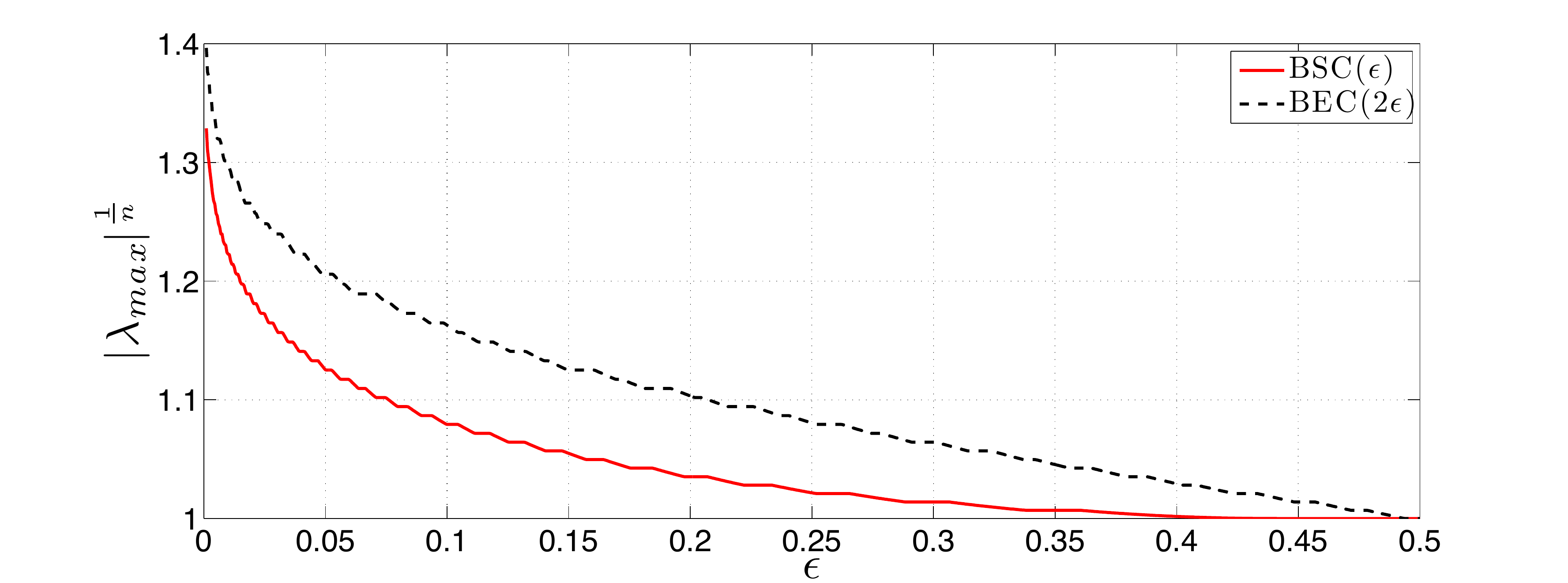}
\caption{Comparing the stabilizable regions of BSC and BEC using linear codes}
\label{fig: BECvBSC}
\end{figure}

\section{Decoding over the BEC}
\label{sec: DecodingBEC}

Owing to the simplicity of the erasure channel, it is possible to come up with an efficient way to perform maximum likelihood decoding at each time step. We will show that the average complexity of the decoding operation at any time $t$ is constant and that it being larger than $KC^3$ decays exponentially in $C$. Consider an arbitrary decoding instant $t$, let $c=[c_1^T,\ldots,c_t^T]^T$ be the transmitted codeword and let $z=[z_1^T,\ldots,z_t^T]^T$ denote the corresponding channel outputs. Recall that $\mathbb{H}_{n,R}^t$ denotes the $\overline{n}t\times nt$ leading principal minor of $\mathbb{H}_{n,R}$. Let $z_e$ denote the erasures in $z$ and let $H_e$ denote the columns of $\mathbb{H}_{n,R}^t$  that correspond to the positions of the erasures. Also, let $\tilde{z}_e$ denote the unerased entries of $z$ and let $\tilde{H}_e$ denote the columns of $\mathbb{H}_{n,R}^t$ excluding $H_e$. So, we have the following parity check condition on $z_e$, $H_ez_e = \tilde{H}_e\tilde{z}_e$. Since $\tilde{z}_e$ is known at the decoder, $s\triangleq \tilde{H}_e\tilde{z}_e$ is known. Maximum likelihood decoding boils down to solving the linear equation $H_ez_e = s$.
Due to the lower triangular nature of $H_e$, unlike in the case of traditional block coding, this equation will typically not have a unique solution, since $H_e$ will typically not be full rank. This is alright as we are not interested in decoding the entire $z_e$ correctly, we only care about decoding the earlier entries accurately. If $z_e = [z_{e,1}^T,\,\, z_{e,2}^T]^T$, then $z_{e,1}$ corresponds to the earlier time instants while $z_{e,2}$ corresponds to the latter time instants. The desired reliability requires one to recover $z_{e,1}$ with an exponentially smaller error probability than $z_{e,2}$. Since $H_e$ is lower triangular, we can write $H_ez_e = s$ as
\begin{align}
\label{eq: bec1}
 \left[\begin{array}{cc}
  H_{e,11} & 0\\
  H_{e,21} & H_{e,22}
 \end{array}\right]\left[\begin{array}{c}z_{e,1}\\z_{e,2}\end{array}\right] = \left[\begin{array}{c}s_1\\s_2\end{array}\right]
\end{align}
Let $H_{e,22}^\bot$ denote the orthogonal complement of $H_{e,22}$, ie., $H_{e,22}^\bot H_{e,22} = 0$. Then multiplying both sides of \eqref{eq: bec1} with diag$(I,H_{e,22})$, we get
\begin{align}
 \label{eq: bec2}
 \left[\begin{array}{c}
  H_{e,11}\\
  H_{e,22}^\bot H_{e,21}
 \end{array}\right]z_{e,1} = \left[\begin{array}{c}s_1\\H_{e,22}^\bot s_2\end{array}\right]
\end{align}
If $[H_{e,11}^T\,\,\, (H_{e,22}^\bot H_{e,21})^T]^T$ has full column rank, then $z_{e,1}$ can be recovered exactly. The decoding algorithm now suggests itself, i.e., find the smallest possible $H_{e,22}$ such that $[H_{e,11}^T\,\,\, (H_{e,22}^\bot H_{e,21})^T]^T$ has full rank and it is outlined in Algorithm \ref{alg: algorithm}.
\begin{algorithm}
\caption{Decoder for the BEC}
\label{alg: algorithm}
\begin{enumerate}
\item Suppose, at time $t$, the earliest uncorrected error is at a delay $d$. Identify $z_e$ and $H_e$ as defined above.
\item Starting with $d'=1,2,\ldots,d$, partition
\begin{align*}
 z_e = [z_{e,1}^T\,\,z_{e,2}^T]^T\,\,\text{and}\,\,H_e = \left[\begin{array}{cc}H_{e,11}&0\\H_{e,21}&H_{e,22}\end{array}\right]
\end{align*}
where $z_{e,2}$ correspond to the erased positions up to delay $d'$.
\item Check whether the matrix $\left[\begin{array}{c}
  H_{e,11}\\
  H_{e,22}^\bot H_{e,21}
 \end{array}\right]$
has full column rank.
\item If so, solve for $z_{e,1}$ in the system of equations
\begin{align*}
  \left[\begin{array}{c}
  H_{e,11}\\
  H_{e,22}^\bot H_{e,21}
 \end{array}\right]z_{e,1} = \left[\begin{array}{c}s_1\\H_{e,22}^\bot s_2\end{array}\right]
\end{align*}
\item Increment $t=t+1$ and continue.
\end{enumerate}
\end{algorithm} 

\subsection{Complexity}
Suppose the earliest uncorrected error is at time $t-d+1$, then steps 2), 3) and 4) in Algorithm \ref{alg: algorithm} can be accomplished by just reducing $H_e$ into the appropriate row echelon form, which has complexity $\OO{d^3}$. The earliest entry in $z_e$ is at time $t-d+1$ implies that it was not corrected at time $t-1$, the probability of which is $P_{d-1,t-1}^e \leq \eta 2^{-n\beta (d-1)}$. Hence, the average decoding complexity is at most $K\sum_{d>0}d^3 2^{-n\beta d}$ which is bounded and is independent of $t$. In particular, the probability of the decoding complexity being $Kd^3$ is at most $\eta 2^{-n\beta d}$. The decoder is easy to implement and its performance is simulated in Section \ref{sec: Simulations}. Note that the encoding complexity per time iteration increases linearly with time. This can also be made constant on average if the decoder can send periodic acks back to the encoder with the time index of the last correctly decoded source bit. 
\section{Simulations}
\label{sec: Simulations}

Consider stabilizing the scalar unstable process of \eqref{eq: system_model} with $\lambda=2$, and $w_t$ and $v_t$ being uniform on $[-30,30]$ and $[-1,1]$ respectively,  over a binary erasure channel with erasure probability $\epsilon = 0.3$. Also, let $n=15$, i.e., we get $15$ channel uses per time instant. To quantize the measurements, we shall use a $L$-regular lattice quantizer with bin width $\delta > V$ which is depicted in Fig \ref{fig: lattice_quant} (see for e.g., \cite{Sahai})
\begin{figure}
\centering
\includegraphics[scale=0.75]{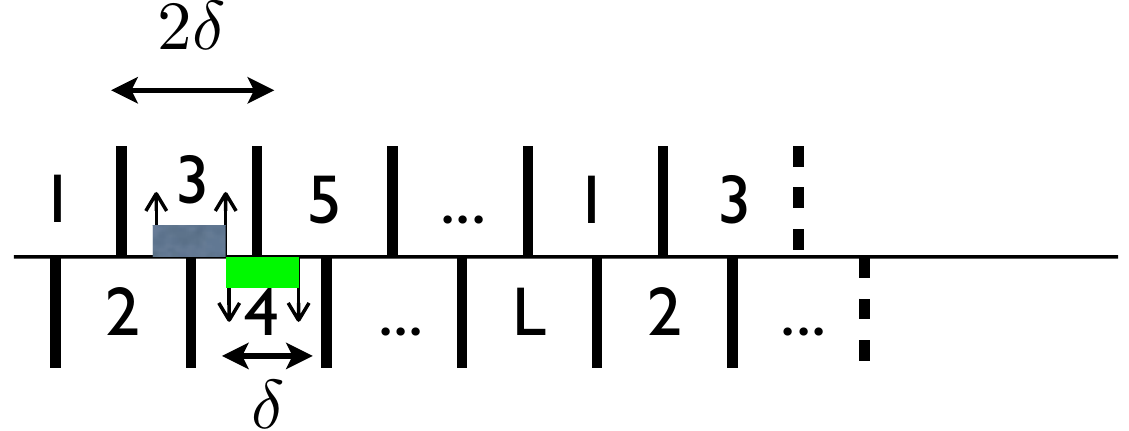}
\caption{$L$-regular lattice quantizer with bin width $\delta$}
\label{fig: lattice_quant}
\end{figure}
If at some time $t$, $x_t\in(-\delta,\delta)$ and the controller is aware of it, then, $y_{t+1}\in \bigl[-\frac{V+W}{2}-|\lambda|\delta, |\lambda|\delta + \frac{W+V}{2}\bigr]$ and hence can lie in atmost one of $\frac{V+W+2|\lambda|\delta}{\delta}$ possible bins. So, if $L > \frac{V+W+2|\lambda|\delta}{\delta}$, then the observer can encode the bin label of $y_{t+1}$ as $b_t\in\{0,1\}^k$, while the channel coder produces the channel inputs $c_t\in\{0,1\}^n$. Clearly we need $2^k \geq L$. So, for $\delta=V$, $L = 35$ and hence $k=6$, the associated code rate $R = 6/15 < 1 - \log_2(1+0.3) = 0.6215$. Using Lemma \ref{lem: Sahai}, inorder to stabilize $|x_t|$, one needs a code with exponent $\beta \geq \frac{1}{n} = 0.0667$. Using Corollary \ref{cor: thresholdsBEC}, causal linear codes exist for $\beta <\beta^*= H^{-1}(1-R)\bra{\log_2\bra{\frac{1}{\bh}} + \log_2\bra{2^{1-R}-1}}$. A quick calculation shows that for $k=6,n=15$, $\beta^* = \frac{1.1413}{n} = 0.0761 > 0.0667$. A time invariant code $\mathbb{H}_{15,0.4}\in\mathbb{TZ}_{\frac{1}{2}}$ was randomly generated and used with the decoding algorithm in Section \ref{sec: DecodingBEC}. Fig \ref{fig: samplePath} shows the plot of a sample path of the above process before and after closing the loop, the fact that the plant has been stabilized is clear.
\begin{figure}[ht]
\centering
\subfigure[Open loop trajectory]{\includegraphics[scale=0.3]{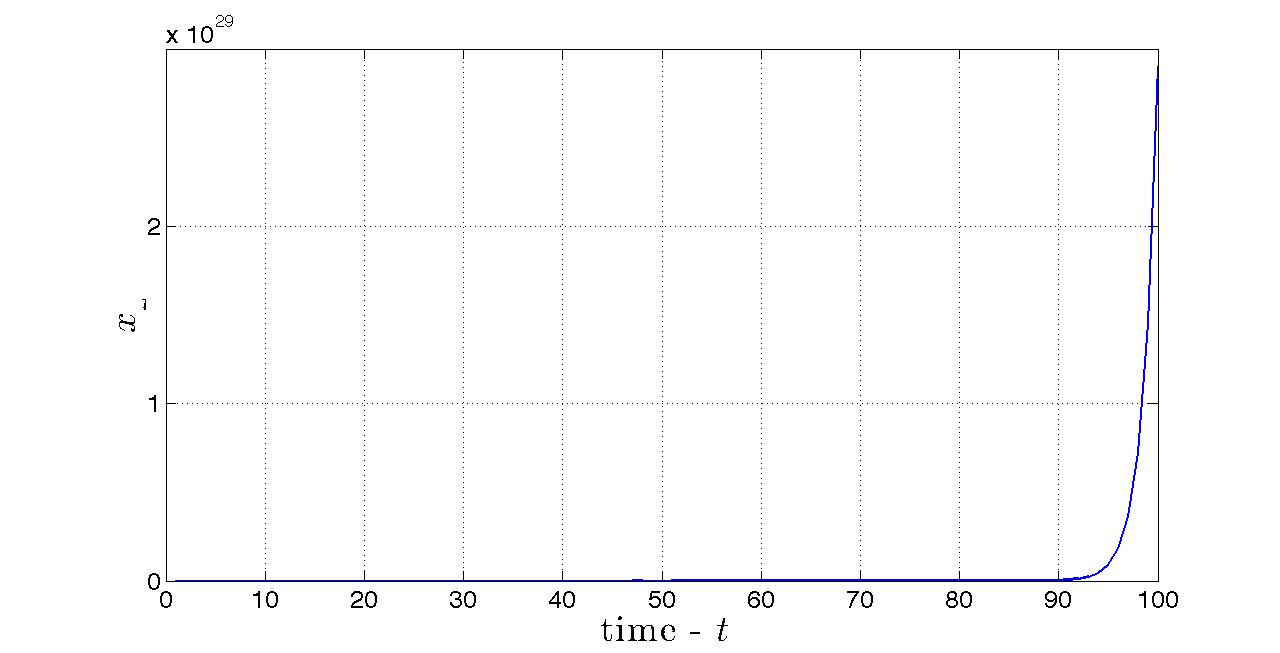}}\,\,
 \subfigure[Trajectory after closing the loop]{\includegraphics[scale=0.3]{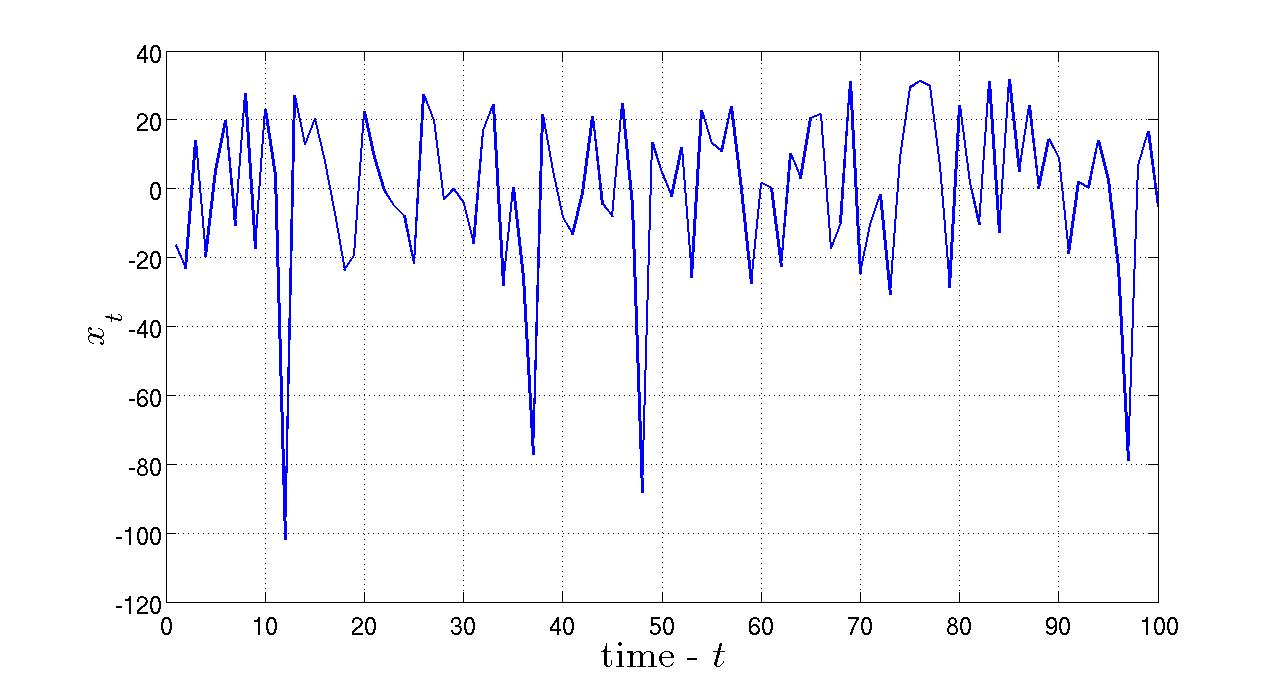}}
\caption{}
\label{fig: samplePath}
\end{figure}

\begin{figure}
 \centering
\includegraphics[scale=0.35]{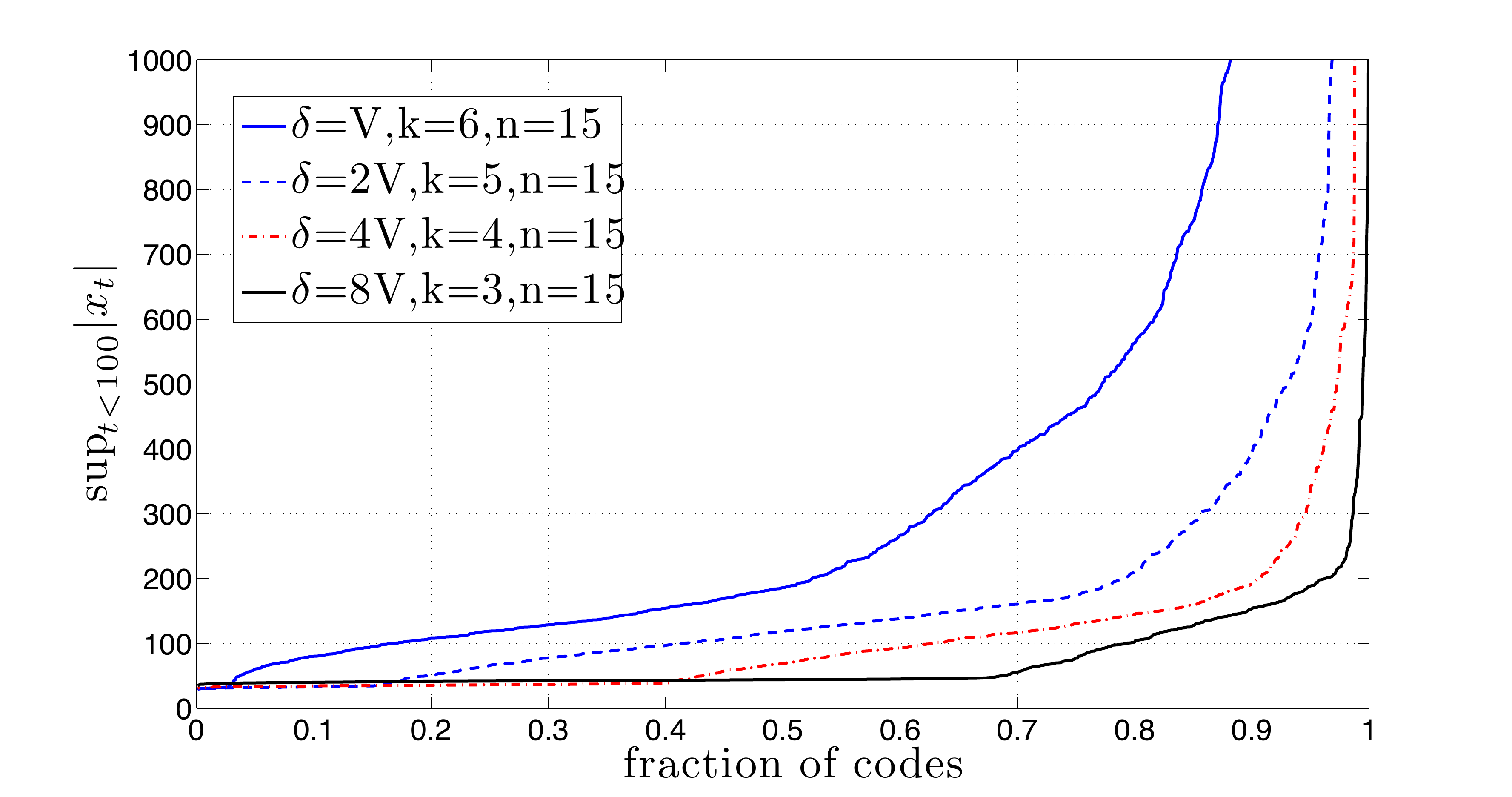}
\caption{The control performance of the code ensemble improves as the rate decreases}
\label{fig: comparison}
\end{figure}
By easing up on the rate $R$, the control performance of the code ensemble improves. This is demonstrated in Fig \ref{fig: comparison}. In the example considered above, one can decrease the rate by increasing $\delta$, e.g., increasing $\delta$ from $V$ to $4V$, $k$ drops from 6 to 4 bits and the resulting rate drops from $0.4$ to $0.267$. For each value of $k$ from 3 to 6, 1000 time invariant codes were generated at random from $\mathbb{TZ}_{\frac{1}{2}}$. Each such code was used to control the process above over a time horizon of $T=100$. The $y-$axis denotes $\sup_{t<100}|x_t|$ while the $x-$axis denotes the proportion of codes for which $\sup_{t<100}|x_t|$ is below a prescribed value. For example, with $k=6,n=15$, $\sup_{t<100}|x_t|$ was less than 200 for $50\%$ of the codes while with $k=3,n=15$, this fraction increases to more than $95\%$. The $y-$axis has been capped at 1000 for clarity. This shows that one can trade-off utilization of communication resources and control performance. A rate $R=0.267$ is much smaller than the capacity $1-\log_2(1+0.3)=0.6215$ but allows codes with much better control performance. 

\section{Conclusion}
\label{sec: Conclusion}

We prove the existence of linear anytime codes with high probability. Our analysis considered binary alphabet, but is easily extendable. This is a significant step, since prior work only demonstrated the existence of such codes. For the BEC, we also propose an efficient decoding algorithm with constant average complexity per iteration, and for which the probability of having a complexity of $KC^3$ at some given iteration decays exponentially in $C$. Simulations validate the efficacy of the method. Constructing code families with efficient decoding for other channels, such as the BSC or AWGN remains an interesting open problem. 
\bibliographystyle{IEEEbib}
\bibliography{ISIT2011arXiv}
\appendix

\begin{lem}
 \label{lem: subspace}
Let $V = \{0,1\}^m$ and define a probability function over $V$ such that, for each $v\in V$, $P(v) = p^{\|v\|}(1-p)^{m-\|v\|}$.
If $U$ is a $\ell-$dimensional subspace of $V$, then
\begin{align*}
 P(U) \leq \begin{cases}
				(1-p)^{m-\ell} & \text{if }\,\,p\leq 1/2 \\
				p^{m-\ell}       & \text{if }\,\,p\geq 1/2
		\end{cases}
\end{align*}
\end{lem}
The proof is easy.
\begin{rem}[Rank of a block lower traingular matrix]
\label{rem: rank}
 Let $A$ be the following block lower triangular matrix
\begin{align*}
   A = \left[\begin{array}{cccc}
				A_{11} & 0 & \ldots & \ldots\\
				A_{21} & A_{22} & 0 & \ldots\\
				\vdots & \vdots & \ddots & \vdots\\
				A_{t1} & A_{t2} & \ldots & A_{tt}\\
				\end{array}\right]
\end{align*}
 If each of $\{A_{jj}\}_{j=1}^t$ has full rank, then $A$ also has full rank.
\end{rem}

\begin{proof}[Proof of Theorem \ref{thm: Toeplitz}]
 Consider an arbitrary decoding instant, $t$. Since $\minW{d}{t} = \minW{d}{t'}$ and $\numW{w}{d}{t} = \numW{w}{d}{t'}$ for all $t,t'$, we will drop these superscripts and write $\minW{d}{t} = \minWi{d}$ and $\numW{w}{d}{t} = \numWi{w}{d}$. Let $c = [c_1^T,\ldots,c_t^T]^T$, where $c_i\in\{0,1\}^n$, be a fixed binary word such that $c_{\tau < t-d+1}=0$ and $c_{t-d+1}\neq 0$. Also, let $\mathbb{H}_{n,R}$ be drawn from the ensemble $\mathbb{TZ}_p$ and let $\mathbb{H}_{n,R}^t$ denote the $\overline{n}t\times nt$ principal minor of $\mathbb{H}_{n,R}$. We examine the probability that $c$ is a codeword of $\mathbb{H}_{n,R}^t$, i.e., $P\bra{\mathbb{H}_{n,R}^tc=0}$. Now, since $c_{\tau<t-d+1}=0$, $\mathbb{H}_{n,R}^tc=0$ is equivalent to
\begin{align}
\label{eq: toeplitzproof1}
\left[\begin{array}{cccc}
	\pc_1    & 0 	    & \ldots      &  \ldots\\
	\pc_2    & \pc_1         & 0 	     & \ldots \\
	\vdots & \vdots      & \ddots     & \vdots\\
	\pc_d     & \pc_{d-1} & \ldots       & \pc_1
\end{array}\right]\left[\begin{array}{c}c_{t-d+1}\\c_{t-d+2}\\\vdots\\c_t\end{array}\right]=\left[\begin{array}{c}0\\0\\\vdots\\0\end{array}\right]
\end{align}
Note that \eqref{eq: toeplitzproof1} can be equivalently written as follows
\begin{align}
\underbrace{ \left[\begin{array}{cccc}
	C_{t-d+1}    & 0 	    & \ldots      &  \ldots\\
	C_{t-d+2}   & C_{t-d+1}       & 0 	     & \ldots \\
	\vdots & \vdots      & \ddots     & \vdots\\
	C_T   & C_{t-1} & \ldots       & C_{t-d+1}
\end{array}\right]}_{\triangleq C}\underbrace{\left[\begin{array}{c}h_1\\h_2\\\vdots\\h_d\end{array}\right]}_{\triangleq h}=\left[\begin{array}{c}0\\0\\\vdots\\0\end{array}\right]
\end{align}
where $h_i = \text{vec}(\pc_i^T)$, i.e., is a $n\overline{n}\times 1$ column obtained by stacking the columns of $\pc_i^T$ one below the other, and $C_i \in\{0,1\}^{\overline{n}\times n\overline{n}}$ is obtained from $c_i$ as follows. 
\begin{align*}
 C_i =  \left[\begin{array}{cccc}
	c_i^T    & 0 	    & \ldots      &  \ldots\\
	0   & c_i^T       & 0 	     & \ldots \\
	\vdots & \vdots      & \ddots     & \vdots\\
	0  & 0 & \ldots       & c_i^T
\end{array}\right]
\end{align*}
Since $c_{t-d+1} \neq 0$, $C_{t-d+1}$ has full rank $\overline{n}$ and consequently $C$ has full rank $d\overline{n}$. So, $P(\mathbb{H}_{n,R}^tc = 0) = P(Ch = 0)$. This is the probability that $h$ lies in a $dn\overline{n} - d\overline{n}$ dimensional null space of $C$. This further implies that $h' = [h_2^T,\ldots,h_d^T]^T$ lies in a $(d-1)n\overline{n} - (d-1)\overline{n}$ dimensional subspace. Note that $h'$ is an $(d-1)n\overline{n}$ dimensional Bernoulli($p$) vector and using Lemma \ref{lem: subspace}, we have
\begin{align}
\label{eq: toeplitzproof2_5}P(\mathbb{H}_{n,R}^tc = 0) &\leq (1-p)^{\overline{n}(d-1)}\\
\implies P\bra{\minWi{d} < \alpha nd} &\leq (1-p)^{\overline{n}(d-1)}\sum_{w'\leq \alpha nd}\binom{nd}{w'}\nonumber\\
 &\leq (1-p)^{\overline{n}(d-1)}2^{ndH(\alpha)}\nonumber\\
\label{eq: toeplitzproof2}& = \eta2^{-nd\bra{(1-R)\log_2(1/(1-p)) - H(\alpha)}}
\end{align}
where $\eta = (1-p)^{-\overline{n}}$. Similarly,
\begin{align}
 \label{eq: teoplitzproof3}
P\bra{\numWi{w}{d} > 2^{\theta w}} &\leq 2^{-\theta w}\E\numWi{w}{d}\nonumber\\
							      &\leq  \eta2^{-\theta w}\binom{nd}{w}(1-p)^{\overline{n}d}\nonumber\\
							      &\leq \eta2^{-nd\bra{\theta w/nd - H(w/nd) + (1-R)\log_2(1/(1-p))}}
\end{align}
For convenience, define
\begin{align*}
 \delta_1 &= (1-R)\log_2(1/(1-p)) - H(\alpha)\\
 \delta_{2,w} &= \theta \frac{w}{nd} - H\bra{\frac{w}{nd}} + (1-R)\log_2(1/(1-p))
\end{align*}
We need to choose $\theta$ such that $\delta_{2,w} > \delta > 0$ for all $\alpha \leq \frac{w}{nd} \leq 1$. Now, define
\begin{align}
 \theta^* = \max_{x\geq \alpha}\frac{H(x) - (1-R)}{x}
\end{align}
Then for each $\theta > \theta^*$, there is a $\delta > 0$ such that $\delta_{2,w} > \delta$ for all $\alpha nd \leq w \leq nd$. A simple calculation gives $\theta^* = \log_2\bra{\frac{1}{2^{1-R}-1}}$. For such a choice of $\theta > \theta^*$, continuing from \eqref{eq: teoplitzproof3}, we have
\begin{align}
\label{eq: toeplitzproof4}
 P\bra{\exists\,\alpha nd\leq w\leq nd\,\,\ni\,\,\numWi{w}{d} > 2^{\theta w}} \leq nd2^{-nd\delta}
\end{align}
for some $\delta' > 0$. For some fixed $d_o$ large enough, applying a union bound over $d\geq d_o$ to \eqref{eq: toeplitzproof2} and \eqref{eq: toeplitzproof4}, we get
\begin{align}
 \label{eq: toeplitzproof5}
P\bra{\exists\,\,d\geq d_o\,\,\ni\,\,\minWi{d}<\alpha nd\text{ or }\numWi{w}{d}>2^{\theta w}} \leq 2^{-\Omega(nd_o)}
\end{align}
Recall that for the above argument to be complete, we also require $\mathbb{H}_{n,R}^t$ to have full rank for each $t$. By Remark \ref{rem: rank}, this is guaranteed if $H_1$ has full rank. So, 
\end{proof} 

\begin{proof}[Proof of Theorem \ref{thm: finiteTimeHorizon}]
 The proof is by induction. Suppose $\mathbb{H}_{n,R,T-1}$ has $(\alpha,\theta,d_o)-$anytime distance. Construct $\mathbb{H}_{n,R,T}$ as follows.
\begin{align*}
 \mathbb{H}_{n,R,T} = \left[\begin{array}{c|ccc}
				\pc_{11} & 0 & \ldots & \ldots \\
				\cline{2-4}\\
				\pc_{21} & & & \\
				\vdots & &\mathbb{H}_{n,R,T-1} &\\
				\pc_{T1} & & &
				\end{array}\right]
\end{align*}
where $\pc_{11}$ is chosen to be a full rank matrix and the entries of $\pc_{j1}\in\{0,1\}^{\overline{n}\times n}$, $j\geq 2$, are drawn according to i.i.d Bernoulli($\frac{1}{2}$). We will show that if $\mathbb{H}_{n,R,T-1}$ has $(\alpha,\theta,d_o)-$anytime distance, then $\mathbb{H}_{n,R,T}$ will is also have $(\alpha,\theta,d_o)-$anytime distance with a probability $1-2^{-\Omega(nd_o)}$. Note that the probability is over the choice of $\{\pc_{j1}\}_{j=1}^T$. Let $\{\minW{d}{t},\numW{w}{d}{t}\}_{d\geq d_o,t\leq T}$ be the weight distribution parameters associated to $\mathbb{H}_{n,R,T}$. Since $\mathbb{H}_{n,R,T-1}$ has $(\alpha,\theta,d_o)-$anytime distance, we have the following
\begin{align*}
 \minW{d}{t} &\geq \alpha nd,\,\,\,\forall\,\,\,d_o\leq d\leq t-1,\,\,t\geq d_o+1\\
 \numW{w}{d}{t} &\leq 2^{\theta w},\,\,\,\forall\,\,\,w\geq\alpha nd,\,\,d_o\leq d\leq t-1,\,\,t\geq d_o+1
\end{align*}
Towards proving that $\mathbb{H}_{n,R,T}$ has $(\alpha,\theta,d_o)-$anytime distance, it remains to show the following holds with a positive probability.
\begin{align}
\label{eq: induction1}
 \text{For } t\geq d_o,\,\,\, \minW{t}{t} \geq \alpha n t,\,\,\,\numW{w}{t}{t} &\leq 2^{\theta w},\,\,\,\forall\,\,\,w\geq \alpha nt
\end{align}
Let $c\in\{0,1\}^{nt}$ such that $c_{\tau < t-d+1}=0$ and $c_{t-d+1}\neq 0$, then it is easy to see that $P\bra{\mathbb{H}_{n,R,T}^tc = 0} = 2^{-n(d-1)}$. The rest of the analysis follows exactly along the lines of the proof of Theorem \ref{thm: Toeplitz} starting from \eqref{eq: toeplitzproof2_5} with $p=\frac{1}{2}$. This gives the following result
\begin{align*}
 &P\bra{\mathbb{H}_{n,R,T}\text{ is bad}|\mathbb{H}_{n,R,T-1}\text{ is good}} = \\
&P\bra{\{\minW{d}{t},\numW{w}{d}{t}\}\text{ do not satisfy \eqref{eq: induction1}}} \leq 1-2^{-\Omega(nd_o)}
\end{align*}
In particular, there exists a choice of $\{\pc_{j1}\}_{j=1}^T$ such that $\mathbb{H}_{n,R,T}$ has $(\alpha,\theta,d_o)-$anytime distance, whenever $\mathbb{H}_{n,R,T-1}$ has $(\alpha,\theta,d_o)-$anytime distance. For the inductive argument to be complete, one needs to prove that there exists a $\mathbb{H}_{n,R,d_o}$ that has $(\alpha,\theta,d_o)-$anytime distance. This is already covered in the proof of the above inductive step.
\end{proof}

\begin{proof}[Proof of Theorem \ref{thm: tighterBSC}]
 The result can be obtained by a slight modification of the proof of Theorem \ref{thm: Toeplitz}. Consider an arbitrary decoding instant $\t$, if the all zero codeword is transmitted is the transmitted codeword, $z=[z_1^T,\ldots,z_t^T]^T$ is the channel output and $v=[v_1,\ldots,v_t^T]^T$ is the bit flip pattern induced by the channel BSC($\epsilon$), $\epsilon < 1/2$, then $z = v$ and for any $c\in\mathcal{C}_{t,d}$, we have
\begin{align}
 P_{t,d}^e &\leq P\bra{\bigcup_{c\in\mathcal{C}_{t,d}}0\text{ decoded as }c,\|v_{t-d+1:t}\|\leq nd(\epsilon+\delta)} \nonumber\\
&+ P\bigl(\|v_{t-d+1:t}\|> nd(\epsilon+\delta)\bigr)\nonumber\\
\label{eq: tighter1}&\leq \sum_{\minWi{d}\leq w\leq 2nd(\epsilon+\delta)}\numWi{w}{d}\bh^w + 2^{-ndKL(\epsilon+\delta\|\epsilon)}
\end{align}
Note that if $\minWi{d} > 2nd(\epsilon+\delta)$, then we get $P_{t,d}^e \leq 2^{-ndKL(\epsilon+\delta\|\epsilon)}$. 
Since we know from Theorem \ref{thm: Toeplitz} that $\minWi{d}$ can scale as $\alpha nd$ for $\alpha$ upto $H^{-1}(1-R)$, for any choice of $R$ such that $H^{-1}(1-R) > 2\epsilon$, there exists a $\delta > 0$ small enough such that $2\epsilon + 2\delta <\alpha < H^{-1}(1-R)$. Fix $\epsilon < 0.25$, then for each rate $R < 1-H(2\epsilon)$, one can achieve an exponent, $\beta$, upto $\beta < KL\bra{\frac{1}{2}H^{-1}(1-R)\| \epsilon}$. It is easy to see that this marks an improvement over Corollary \ref{cor: thresholdsBEC} only for $\epsilon$ such that $1-H(2\epsilon) > 1-2\log_2(\sqrt{\epsilon}+\sqrt{1-\epsilon})$. The result is now immediate. As an aside, a quick numerical calculation gives $\epsilon^* \approx 0.0753$. 
\end{proof}

\end{document}